\documentclass[11pt]{article}

\usepackage[blocks]{authblk}
\usepackage{stmaryrd}
\usepackage{amsmath}
\usepackage{amssymb}
\usepackage{amsfonts}
\usepackage{algorithm2e}
\usepackage[a4paper,margin=1in,footskip=.5cm]{geometry}
\usepackage[usenames,dvipsnames]{pstricks}
\usepackage{epsfig}
\usepackage{pst-grad}
\usepackage{pst-plot}
\usepackage{mathtools}
\usepackage{float}

\def \qed {\hfill $\Box$}

\def \mbZ {\mathbb{Z}_p}
\def \mbZs {\mathbb{Z}_p^*}
\def \mcI {\mathcal{I}}

\newtheorem{theorem}{\bf Theorem}[section]
\newtheorem{conjecture}{\bf Conjecture}[section]
\newtheorem{proposition}{\bf Proposition}[section]
\newtheorem{lemma}{\bf Lemma}[section]

\newtheorem{remark}{\bf Remark}[section]

\title{On Acyclic Edge-Coloring of Complete Bipartite Graphs}

\author[1]{Ayineedi Venkateswarlu\thanks{{e-mail: \texttt{venku@isichennai.res.in}}; Corresponding author}}
\affil[1]{Computer Science Unit, Indian Statistical Institute - Chennai Centre,
MGR~Knowledge~City~Road, Taramani, Chennai~--~600113, INDIA.}
\author[2]{Santanu Sarkar\thanks{{e-mail: \texttt{santanu@iitm.ac.in}}}}
\affil[2]{Department of Mathematics\\ Indian Institute of Technology Madras,
Chennai~--~600036, INDIA.}
\author[1]{A. Sai Mali\thanks{{e-mail: \texttt{sai.mali.mail@gmail.com }}}}
\date{}

\begin{document}
\maketitle

\begin{abstract}
An acyclic edge-coloring of a graph is a proper edge-coloring without bichromatic
($2$-colored) cycles.  The acyclic chromatic index of a graph $G$, denoted by $a'(G)$,
is the least integer $k$ such that $G$ admits an acyclic edge-coloring using $k$ colors.
Let $\Delta = \Delta(G)$ denote the maximum degree of a vertex in a graph $G$.
A complete bipartite graph with $n$ vertices on each side is denoted by $K_{n,n}$.
Basavaraju, Chandran and Kummini proved that $a'(K_{n,n}) \ge n+2 = \Delta + 2$
when $n$ is odd. Basavaraju and Chandran provided an acyclic edge-coloring of $K_{p,p}$
using $p+2$ colors and thus establishing $a'(K_{p,p}) = p+2 = \Delta + 2$ when $p$
is an odd prime. The main tool in their approach is perfect $1$-factorization of $K_{p,p}$.
Recently, following their approach, Venkateswarlu and Sarkar have shown that $K_{2p-1,2p-1}$
admits an acyclic
edge-coloring using $2p+1$ colors which implies that $a'(K_{2p-1,2p-1}) = 2p+1 = \Delta + 2$,
where $p$ is an odd prime. In this paper, we generalize this approach and present a general
framework to possibly get an acyclic edge-coloring of $K_{n,n}$ which possess a perfect
$1$-factorization  using $n+2 = \Delta+2$ colors. In this general framework, we
show that $K_{p^2,p^2}$ admits an acyclic edge-coloring using $p^2+2$ colors and thus
establishing $a'(K_{p^2,p^2}) = p^2+2 = \Delta + 2$ when $p\ge 5$ is an odd prime.\\

\end{abstract}

\section{Introduction}
Let $G=(V,E)$ be a finite and simple graph. A \emph{proper edge-coloring} of $G$ is an
assignment of colors to the edges so that no two adjacent edges have same color. So it is a map 
$\theta: E \rightarrow \mathcal{C}$ with $\theta(e) \neq \theta(f)$ for any adjacent edges $e,f\in E$,
where $\mathcal{C}$ is the set of colors. The \emph{chromatic index}, denoted by $\chi'(G)$,
is the least integer $k$ such that $G$ admits a proper edge-coloring using $k$ colors.
A proper coloring of $G$ is \emph{acyclic} if there is no two colored cycle in $G$. 
The \emph{acyclic edge chromatic number} (also called \emph{acyclic chromatic index}),
denoted by $a'(G)$, is the least integer $k$ such that $G$ admits an acyclic edge-coloring
using $k$ colors. The notion of acyclic coloring was first introduced by Gr{\"{u}}nbaum~\cite{gru}
in~$1973$, and the concept of acyclic edge-coloring was first studied by~Fiam\u{c}\'{i}k~\cite{fl}.
Let $\Delta=\Delta(G)$ be the maximum degree of a vertex in $G$. It is obvious that
any proper edge-coloring requires at least $\Delta$ colors. Vizing~\cite{vi} proved that
there always exists a proper edge-coloring with $\Delta+1$ colors. Since any acyclic edge
coloring is proper, we must have $a'(G) \geq \chi'(G) \geq \Delta$. On the other hand,
in $1978$, Fiam\u{c}\'{i}k~\cite{fl} (also independently by Alon, Sudakov and Zaks~\cite{al2})
posed the following conjecture.
\begin{conjecture}\label{conj}
for any graph $G,\ a'(G) \leq \Delta+2$.
\end{conjecture}

In~\cite{al2} it was proved that there exists a constant $c$ such that $a'(G) \leq \Delta +2$
for any graph with girth is at least $c\Delta\log \Delta$. It was also proved in~\cite{al2} that
$a'(G) \leq \Delta+2$ for almost all $\Delta$-regular graphs. Later N{\v{e}}set{\v{r}}il
and Wormald~\cite{nw} improved 
this bound and showed that $a'(G) \leq \Delta +1$ for a random regular graph $G$. In another
direction, there have been many results giving upper bounds on $a'(G)$ for an arbitrary graph $G$.
For example, Alon, McDiarmid and Reed~\cite{al} proved that $a'(G) \leq 64 \Delta$. Molloy and
Reed~\cite{mr} improved this bound and showed that $a'(G) \leq 16 \Delta$. Recently,
Ndreca~{\em et. al.} obtained $a'(G) \leq 9.62 \Delta$~\cite{ndr} which is currently the best
upper bound for an arbitrary graph $G$.
Muthu, Narayanan and Subramanian~\cite{m1,m4} obtained better bounds: $a'(G) \leq 4.52 \Delta$
for graphs with girth at least 220; $a'(G) \leq 6\Delta$ for graphs with girth at least $9$.
The acyclic edge-coloring of planar graphs has been deeply studied in recent years.
See~\cite[Section~3.3]{wb} for a nice account of recent results.

The Conjecture~\ref{conj} was shown to be true for some special classes of graphs.
Burnstein~\cite{ber} showed that $a'(G) \leq 5$  when $\Delta=3$. Hence the
conjecture is true when $\Delta\leq 3$.  Muthu, Narayanan and Subramanian proved
that the conjecture holds true for grid-like graphs~\cite{m2} and outerplanner graphs~\cite{m3}. 
It has been observed that determining $a'(G)$ is a hard problem from both theoretical
and algorithmic points of view~\cite[p. 2119]{wb}. In fact, we do not yet know the values of
$a'(G)$ for some simple and highly structured graphs like complete graphs and complete bipartite
graphs in general. Fortunately, we can get the exact value of $a'(G)$ for some cases of complete
bipartite graphs, thanks to the perfect $1$-factorization.

Let $K_{n,n}$ be the complete bipartite graph with $n$ vertices on each side. The complete
bipartite graph $K_{n,n}$ is said to have a perfect 1-factorization if  the edges of $K_{n,n}$
can be decomposed into $n$ disjoint perfect matchings such that the union of any two
perfect matchings gives a Hamiltonian cycle and it is of length $2n$
(see Section~\ref{sec2} for more details).  It is known that when $n+2 \in \{p, 2p-1,p^2\}$,
where $p$ is an odd prime, or $n+2<50$ and odd, then $K_{n+2,n+2}$ has a
perfect 1-factorization (see~\cite{br}). One can easily see that if $K_{n+2,n+2}$
has a perfect 1-factorization then $a'(K_{n,n}) \leq a'(K_{n+1,n+1}) \leq n+2$.
And also we have the following result due to Basavaraju, Chandran and Kummini~\cite{ba}.
\begin{theorem}\label{thm1}
$a'(K_{n,n}) \geq n+2=\Delta+2$, when $n$ is odd.
\end{theorem}
Hence $a'(K_{n,n})=n+2=\Delta+2$ when $n+2 \in \{p, 2p-1,p^2\}$. By a result of
Guldan~\cite[Corollary 1]{gul}, we can also get $a'(K_{n+1,n+1})=n+2=\Delta+1$
when $n+2 \in \{p, 2p-1,p^2\}$.

The main idea here is to give different colors to the edges in different $1$-factors in $K_{n+2,n+2}$,
and removal of (one) two vertices on each side and their associated edges gives the required
edge-coloring of $(K_{n+1,n+1})\,K_{n,n}$ with $n+2$ colors. But a different approach is
needed to deal with $K_{n+2,n+2}$ when $n+2 \in \{p, 2p-1, p^2\}$. In 2009, Basavaraju
and Chandran~\cite{ba2} proved that $a'(K_{p,p}) = p+2=\Delta+2$ for any odd prime $p$.
We can view their approach as follows: suitably pick one edge from each $1$-factor and partition
these edges into two groups and each group can possibly be assigned a different color to
get the required result. Following this approach, Venkateswarlu and Sarkar have recently
shown that $a'(K_{2p-1,2p-1}) = 2p+1=\Delta+2$ for any odd prime $p$~\cite{AVSS}.
In this paper we view this approach in a more general setting and propose a general
framework for the proof.  The only remaining infinite class of complete bipartite graphs
that are known to have perfect $1$-factorization is $K_{p^2,p^2}$, where $p$ is odd
prime. In this general framework we provide an acyclic edge-coloring of $K_{p^2,p^2}$
using $p^2+2$ colors when $p\ge 5$. Therefore we state our main result as follows.
\begin{theorem}\label{thm2}
$a'(K_{p^2,p^2}) = p^2+2=\Delta+2$, where $p\ge 5$ is an odd prime.
\end{theorem}

Therefore the acyclic chromatic index is equal to $\Delta+2$ for all the three known infinite
classes of complete bipartite graphs having a perfect $1$-factorization,  and the Conjecture~\ref{conj}
holds true for such graphs.

In the next section we discuss some preliminaries and in Section~\ref{sec3} we present
a general framework to possibly get an acyclic edge-coloring of $K_{n,n}$ which possess
a perfect $1$-factorization using $n+2$ colors. Then we present the proof of
Theorem~\ref{thm2} in this framework in Section~\ref{sec4}.

\section{Preliminaries}\label{sec2}
Let $n\,(\geq 2)$ be an integer. We treat elements of the ring $\mathbb{Z}_n$ as integers in
the range $\{0,1,\ldots,n-1\}$. We denote the complete bipartite graph $K_{n,n}$ as
$G = (V\cup V', E)$ with $|V| = |V'| = n$ and $E = \{(v\mapsto v') : v\in V\ \mbox{and}\ v'\in V'\}$. 
We use $\mapsto$ to define edges though our graph $K_{n,n}$ is undirected. 
This is only for ease of presentation in associating a perfect matching in $K_{n,n}$ with
a permutation of the label set $I (= \{0,1,\ldots,n-1\})$, which we discuss below. Accordingly,
the use of arrows in the Figures~\ref{fig1} and~\ref{fig2} below is to explicitly emphasize
the correspondence between a perfect matching and its associated permutation map.
We use the terms `composition' and `product' of permutations interchangeably.
Note also that a permutation can be decomposed as a product of disjoint cycles uniquely
(upto reorder of the cycles and cyclic rotation of the elements within a cycle) and it is called
disjoint cycle decomposition. We use $\sqcup$ (instead of the usual union notation $\cup$)
to signify union of `disjoint' sets.
\subsection{Perfect matching and Perfect $1$-factorization}\label{sec2p1}
A \emph{matching} in a graph is a set of edges without common vertices, and a \emph{perfect matching}
is a matching which matches all vertices of the graph. In the case of complete bipartite graph $K_{n,n}$,
a perfect matching $M\subset E$ is a set of $n$ edges satisfying:
\begin{itemize}
\item[-] for each vertex $v'\in V'$ there exists a vertex $v\in V$ such that $(v\mapsto v')\in M$.
\item[-] if $(v_1\mapsto v')$ and $(v_2\mapsto v')$ are in $M$ then $v_1 = v_2$.
\end{itemize}

So by labelling the vertices in both $V$ and $V'$ with elements of $I = \{0,1,\ldots,n-1\}$
(or an appropriate label set $I$ of size $n$),
we can interpret a perfect matching $M$ in $K_{n,n}$ as a permutation of the label set $I$,
say $\pi_M$. For convenience, let us illustrate this through an example. Let $n=5$ and
consider the graph $K_{5,5}$ with the same labels from $0$ to $4$ for the vertices
on the top $(V)$ and the bottom $(V')$, as depicted in the figure below.
Let $M = \{(0\mapsto 1), (1\mapsto 2),(2\mapsto 0), (3\mapsto 4),(4\mapsto 3)\}$
then $\pi_M = (012)(34)$.

\begin{figure}[ht]
\centering{
\begin{pspicture}(-0.20,-0.80)(3.80,0.80)
\psline[linecolor=green]{->}(0.00,0.60)(0.85,-0.50)
\psline[linecolor=green]{->}(0.90,0.60)(1.75,-0.50)
\psline[linecolor=green]{->}(1.80,0.60)(0.05,-0.54)
\psline[linecolor=green]{->}(2.70,0.60)(3.55,-0.50)
\psline[linecolor=green]{->}(3.60,0.60)(2.75,-0.50)
\put(-0.10,0.72){$0$}
\put(0.80,0.72){$1$}
\put(1.70,0.72){$2$}
\put(2.60,0.72){$3$}
\put(3.50,0.72){$4$}
\put(-0.10,-0.98){$0$}
\put(0.80,-0.98){$1$}
\put(1.70,-0.98){$2$}
\put(2.60,-0.98){$3$}
\put(3.50,-0.98){$4$}
\psdots[dotsize=0.12](0.00,0.60)
\psdots[dotsize=0.12](0.90,0.60)
\psdots[dotsize=0.12](1.80,0.60)
\psdots[dotsize=0.12](2.70,0.60)
\psdots[dotsize=0.12](3.60,0.60)
\psdots[dotsize=0.14,dotstyle=o](1.80,-0.60)
\psdots[dotsize=0.14,dotstyle=o](0.90,-0.60)
\psdots[dotsize=0.14,dotstyle=o](2.70,-0.60)
\psdots[dotsize=0.14,dotstyle=o](0.00,-0.60)
\psdots[dotsize=0.14,dotstyle=o](3.60,-0.60)
\end{pspicture}
\caption{$\pi_M = (012)(34)$}\label{fig1}
}
\end{figure}
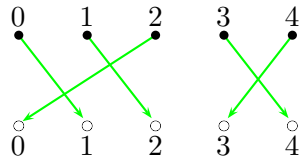 

Note that the union of any two perfect matchings of $K_{n,n}$ forms a collection of disjoint
cycles. These cycles can also be seen from the disjoint cycle decomposition of the composition
of their corresponding permutations. Let $M$ be as mentioned above and
$M' = \{(0\mapsto 1), (1\mapsto 0),(2\mapsto 2), (3\mapsto 3),(4\mapsto 4) \}$.
Then $\pi_{M'} = (01)(2)(3)(4)$.

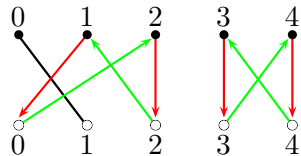
\begin{figure}[ht]
\centering{
\begin{pspicture}(-0.20,-0.80)(3.80,0.80)
\psline[linewidth=0.03cm](0.02,0.56)(0.86,-0.56)
\psline[linecolor=green]{<-}(0.95,0.55)(1.80,-0.60)
\psline[linecolor=green]{<-}(1.75,0.55)(0.00,-0.60)
\psline[linecolor=green]{<-}(2.75,0.55)(3.60,-0.60)
\psline[linecolor=green]{<-}(3.55,0.55)(2.70,-0.60)
\psline[linecolor=red]{->}(0.90,0.60)(0.00,-0.50)
\psline[linecolor=red]{->}(1.80,0.60)(1.80,-0.50)
\psline[linecolor=red]{->}(2.70,0.60)(2.70,-0.50)
\psline[linecolor=red]{->}(3.60,0.60)(3.60,-0.50)
\put(-0.10,0.72){$0$}
\put(0.80,0.72){$1$}
\put(1.70,0.72){$2$}
\put(2.60,0.72){$3$}
\put(3.50,0.72){$4$}
\put(-0.10,-0.98){$0$}
\put(0.80,-0.98){$1$}
\put(1.70,-0.98){$2$}
\put(2.60,-0.98){$3$}
\put(3.50,-0.98){$4$}
\psdots[dotsize=0.12](0.00,0.60)
\psdots[dotsize=0.12](0.90,0.60)
\psdots[dotsize=0.12](1.80,0.60)
\psdots[dotsize=0.12](2.70,0.60)
\psdots[dotsize=0.12](3.60,0.60)
\psdots[dotsize=0.14,dotstyle=o](1.80,-0.60)
\psdots[dotsize=0.14,dotstyle=o](0.90,-0.60)
\psdots[dotsize=0.14,dotstyle=o](2.70,-0.60)
\psdots[dotsize=0.14,dotstyle=o](0.00,-0.60)
\psdots[dotsize=0.14,dotstyle=o](3.60,-0.60)
\end{pspicture}
\caption{Induced subgraph of $M \cup M'$ (in $K_{5,5}$)}\label{fig2}
}
\end{figure} 
We have $\pi_M^{-1}\circ \pi_{M'} = (0) (12) (34)$ and the arrows are placed
accordingly in the figure above. The fixed element $(0)$ corresponds to the common
edge $(0\mapsto 1) \in M\cap M'$ represented by  the normal line in Figure~\ref{fig2},
and the two cycles $C_1 = (12)$ and $C_2 = (34)$ correspond to the cycles
$C_1^g = \{(1\mapsto 0),(2\mapsto 0),(2\mapsto 2),(1\mapsto 2)\}$ and
$C_2^g = \{(3\mapsto 3),(4\mapsto 3),(4\mapsto 4),(3\mapsto 4)\}$ respectively in $K_{5,5}$.
By a careful observation of the example, we can see the following general result.
\begin{lemma}\label{lem01}
Let $M$ and $M'$ be two perfect matchings of $K_{n,n}$ and let $C= (i_0\,i_1\cdots i_{\ell-1}),\,\ell\geq 2,$
be a cycle of length $\ell$ in the disjoint cycle decomposition of $\pi_M^{-1}\circ \pi_{M'}$,
{\em i.e.}, $\pi_{M'}(i_j) = \pi_{M}(i_{j+1})$, where the subscripts
are taken modulo $\ell$. Then the corresponding cycle $C^g$ in $K_{n,n}$ is of length $2\ell$
and the participating edges are given by $\{(i_j\mapsto \pi_M(i_j)) : 0\le j\le \ell-1\}\subset M$ and
$\{(i_j\mapsto \pi_{M'}(i_j)) : 0\le j\le \ell-1\}\subset M'$ which appear alternatively in $C^g$
as depicted in Figure~\ref{fig3}\,: the edges of $M'$ and $M$ are represented by lines
(without arrows) with colors red and green respectively.
\end{lemma}
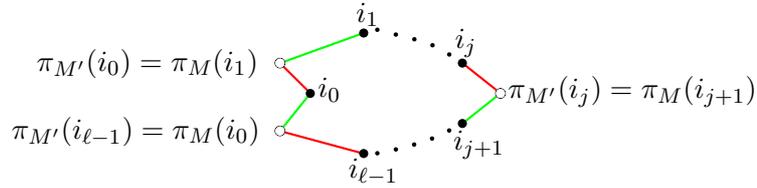
\begin{figure}[ht]
\centering{
\scalebox{1} 
{
\begin{pspicture}(-4.00,-1)(4.00,1.2)
\pscurve[linewidth=0.06,linestyle=dotted,dotsep=0.2cm](0.85,1.05)(1.12,0.98)(1.6,0.84)(1.85,0.70)
\pscurve[linewidth=0.06,linestyle=dotted,dotsep=0.2cm](0.85,-0.60)(1.1,-0.55)(1.6,-0.40)(1.85,-0.28)
\psline[linecolor=red](0.70,-0.60)(-0.40,-0.30)
\psline[linecolor=green](-0.00,0.20)(-0.40,-0.30)
\psdots[dotsize=0.14,dotstyle=o](-0.40,-0.30)
\psline[linecolor=green](0.70,1.00)(-0.40,0.60)
\psline[linecolor=red](-0.00,0.20)(-0.40,0.60)
\psdots[dotsize=0.14,dotstyle=o](-0.40,0.60)
\psline[linecolor=red](2.00,0.60)(2.50,0.20)
\psline[linecolor=green](2.00,-0.20)(2.50,0.20)
\psdots[dotsize=0.14,dotstyle=o](2.50,0.20)
\psdots[dotsize=0.12](0.00, 0.20)
\psdots[dotsize=0.12](0.70, 1.00)
\psdots[dotsize=0.12](2.00, 0.60)
\psdots[dotsize=0.12](2.00,-0.20)
\psdots[dotsize=0.12](0.70,-0.60)
\put(0.10,0.15){$i_0$}
\put(0.60,1.15){$i_1$}
\put(0.50,-0.92){$i_{\ell-1}$}
\put(1.90,0.80){$i_{j}$}
\put(1.88,-0.55){$i_{j+1}$}
\put(-3.60,0.50){$\pi_{M'}(i_0)=\pi_M(i_1)$}
\put(-3.93,-0.40){$\pi_{M'}(i_{\ell-1})=\pi_M(i_0)$}
\put(2.60,0.14){$\pi_{M'}(i_{j})=\pi_M(i_{j+1})$}
\end{pspicture} 
}
\caption{Cycle $C^g$ in the induced subgraph of $M \cup M'$ (in $K_{n,n}$)}\label{fig3}
}
\end{figure} 

A perfect matching is also called a \emph{$1$-factor}, and a partitioning of the edges of a
graph into $1$-factors is a \emph{$1$-factorization}. A $1$-factorization is \emph{perfect}
if the union of any two of its $1$-factors (perfect matchings) is a Hamiltonian cycle. 
As pointed out in the introduction, there are three infinite classes of complete bipartite graphs
known to have perfect $1$-factorization, namely, $n\in \{p, 2p-1, p^2\}$ with $p$
an odd prime. Let us illustrate it by considering the complete bipartite graph $K_{5,5}$.
As discussed above the $1$-factors can be described by
permutations of $I=\{0,1,2,3,4\} $. Consider the $1$-factors $M_i$ given by the permutations
$\pi_i(j) = (i+j)\bmod{5}$ for $i,j\in \mathbb{Z}_5$. In Figure~\ref{fig4} the edges of $M_1$
are the green colored lines and the edges of $M_2$ are the red colored lines and they correspond to
the permutations $\pi_1  = (01234)$ and $\pi_2 = (02413)$ respectively. The induced subgraph
formed by the edges $M_1 \sqcup M_2$ is depicted below (without arrows), and the corresponding
permutation is equal to $\sigma = \pi_1^{-1}\circ \pi_2 = (01234)$. 

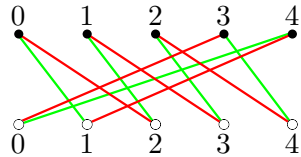
\begin{figure}[ht]
\centering{
\begin{pspicture}(-0.20,-0.80)(3.80,0.80)
\psline[linewidth=0.03cm,linecolor=green](0.00,0.58)(0.90,-0.60)
\psline[linewidth=0.03cm,linecolor=green](0.90,0.58)(1.80,-0.60)
\psline[linewidth=0.03cm,linecolor=green](1.80,0.58)(2.70,-0.60)
\psline[linewidth=0.03cm,linecolor=green](2.70,0.60)(3.60,-0.60)
\psline[linewidth=0.03cm,linecolor=green](3.60,0.62)(0.00,-0.60)
\psline[linewidth=0.03cm,linecolor=red](0.00,0.60)(1.80,-0.60)
\psline[linewidth=0.03cm,linecolor=red](0.90,0.60)(2.70,-0.60)
\psline[linewidth=0.03cm,linecolor=red](1.80,0.60)(3.60,-0.60)
\psline[linewidth=0.03cm,linecolor=red](2.70,0.60)(0.00,-0.57)
\psline[linewidth=0.03cm,linecolor=red](3.60,0.60)(0.90,-0.59)
\put(-0.10,0.72){$0$}
\put(0.80,0.72){$1$}
\put(1.70,0.72){$2$}
\put(2.60,0.72){$3$}
\put(3.50,0.72){$4$}
\put(-0.10,-0.98){$0$}
\put(0.80,-0.98){$1$}
\put(1.70,-0.98){$2$}
\put(2.60,-0.98){$3$}
\put(3.50,-0.98){$4$}
\psdots[dotsize=0.12](0.00,0.60)
\psdots[dotsize=0.12](0.90,0.60)
\psdots[dotsize=0.12](1.80,0.60)
\psdots[dotsize=0.12](2.70,0.60)
\psdots[dotsize=0.12](3.60,0.60)
\psdots[dotsize=0.14,dotstyle=o](1.80,-0.60)
\psdots[dotsize=0.14,dotstyle=o](0.90,-0.60)
\psdots[dotsize=0.14,dotstyle=o](2.70,-0.60)
\psdots[dotsize=0.14,dotstyle=o](0.00,-0.60)
\psdots[dotsize=0.14,dotstyle=o](3.60,-0.60)
\end{pspicture}
\caption{Induced subgraph of $M_1 \sqcup M_2$ (in $K_{5,5}$)}\label{fig4}
}
\end{figure} 

Now we can see from the above lemma that $M_0, M_1,\ldots, M_{n-1}$
is a perfect 1-factorization of the complete bipartite graph $K_{n,n}$ if and only if
the permutation $(\pi_i^{-1}\circ \pi_j)$ is a full cycle (of length $n$) for any
$0\le i, j\le(n-1)$ with $i\neq j$, where $\pi_i$ is the permutation corresponding to
$M_i$. Next we present a general framework to possibly get an acyclic edge-coloring of
$K_{n,n}$ which possess a perfect $1$-factorization using $n+2$ colors, where $n$ is odd.

\section{A general  framework}\label{sec3}
Let $n\,(\ge 3)$ be an odd integer. Suppose that the complete bipartite graph $K_{n,n}$
has a perfect 1-factorization. Then there exist $n$  disjoint perfect matchings covering
all the $n^2$ edges of $K_{n,n}$, say $M_0, \ldots, M_{n-1}$, such that the union of any
two perfect matchings $M_i \sqcup M_j$ gives a Hamiltonian cycle (which is of length $2n$).
Let $M_i^*\subsetneqq M_i$ be a proper subset of $M_i$ and consider the
following partial coloring:
\begin{equation}\label{eqn01}
\mbox{assign\ color}\ c_i\ \mbox{to}\ \mbox{the\ edges\ in}\ M_i^{*}.
\end{equation}
The remaining edges to be colored is given by
\[ M = \bigsqcup_{i=0}^{n-1} (M_i\setminus M_i^*),\]
and these edges will be assigned some other colors different from the colors in $\{c_0,\ldots,c_{n-1}\}$.
\begin{lemma}\label{lem00}
There can not be a cycle in the induced subgraph formed by the edges from the union
$M_i^*\sqcup M_j^*$ of two color classes $c_i$ and $c_j$ for any $0\le i, j\le (n-1)$ with $i\neq j$.
\end{lemma}

Ideally, to get an acyclic edge-coloring of $K_{n,n}$, our aim should be to use only two
more colors for coloring the (uncolored) edges in $M$, and thus attaining the lower
bound on $a'(K_{n,n})$. In other words, partition $M$ into $M^{(1)}$ and $M^{(2)}$,
if possible, in such a way that the induced subgraph of $M^{(\rho)}\sqcup M_i^*$
does not contain a cycle for any $0\leq i \leq n-1$ and $\rho\in\{1,2\}$. If such a partition
of $M$ exists, then the edges of $M^{(1)}$ and $M^{(2)}$ can be assigned distinct colors,
and one can easily see that the proposed edge-coloring is proper and acyclic.
These observations and our intuition suggest that to minimize the number of edges (size of $M$)
that are yet to be assigned colors. This can be done by taking $M_i^*$ to be a (proper) maximal
subset of $M_i$, {\em i.e.}, take $M_i^* = M_i \setminus \{ e_i\}$ for some $e_i \in M_i$.
Thus the size of the set $M$ will be (near)\footnote{One can take $M_i^*= M_i$
for exactly one $i\in I$ and other $M_i^*$'s as mentioned, and observe that
Lemma~\ref{lem00} holds true in such a case as well.} minimal.
Additionally, we have to make sure that a suitable partition of $M$ exists
satisfying the other requirements mentioned above.  Below we present a strategy to
choose $M$ (to be a perfect matching), so that the other requirements on $M$ can
possibly be worked out using permutations of the label set $I$.

As discussed in the previous section, we can interpret each perfect matching $M_i$
in $K_{n,n}$ as a permutation of the label set $I$ and let us denote it by $\pi_i$. 
Now, if possible, select one edge $e_i : (s_i\mapsto t_i)$ from each
$M_i,\,0 \leq i \leq (n-1),$ such that the set of edges $M = \{e_0, e_1, \ldots, e_{n-1}\}$
gives a perfect matching in $K_{n,n}$. Let $\pi$ be its corresponding permutation of $I$.
Then we have $M \cap M_i = \{ e_i\}$ for $1\le i\le (n-1)$.

\begin{remark}\label{rem_trans}
A perfect $1$-factorization of complete bipartite graph $K_{n,n}$ is equivalent to
a Hamiltonian latin square of order $n$. Our choice of perfect matching $M$ with
$M \cap M_i = \{ e_i\}$ for $1\le i\le (n-1)$ is equivalent to a transversal (of length
$n$). These concepts are well studied in literature (see~\cite[p. 135--151]{CD}).
It was conjectured by Ryser that every latin square of odd order has a transversal~\cite{brua}.
Such a transversal is also said to be rainbow matching (see~\cite{aha} for details).
In fact, it is enough to take a transversal of length $n-1$ as pointed out in the footnote
below. It was conjectured by Brualdi and Stein independently that every latin square of odd
order $n$ has a transversal of length $n-1$~\cite{brua,stein}. As far as we know it is not
known that in general a Hamiltonian latin square of order $n$ has a transversal (of length $n$ or $n-1$).
In fact, we need a transversal which satisfies an additional property as discussed below.
\end{remark}

Now set $M_i^* = M_i \setminus \{ e_i\}$ and assign colors as in~(\ref{eqn01}).
As discussed our aim is to use two additional colors for coloring the remaining edges
(given by $M$) and obtain an acyclic edge-coloring of $K_{n,n}$. For this purpose,
we need to partition $M$ into $M^{(1)}$ and $M^{(2)}$, in such a way that the induced
subgraph of $M^{(\rho)} \cup M_i^*$ does not contain a cycle for any $0 \leq i \leq n-1$ and
$\rho\in\{1,2\}$. In other words, we must have edges from both $M^{(1)}$ and $M^{(2)}$
in any cycle in the induced subgraph of $M\sqcup M_i^* = M\cup M_i$. Note that by
Lemma~\ref{lem01} the cycles of the induced subgraph of $M\cup M_i$ can be obtained
from the cycles of the permutation $\pi^{-1}\circ \pi_i$ in its disjoint cycle decomposition.
So in order to see such a partition of $M$ exists or not, we analyse cycle structure
of the permutations $(\pi^{-1}\circ \pi_i)$ for $0\le i\le (n-1)$. We now see that
such a partition of $M$, if exists, can be obtained from a partition of the label set $I$, and
it is due to the one-to-one correspondence between the label set and a perfect matching.

Suppose $(\pi^{-1}\circ \pi_i) = C_{i0} C_{i1} \cdots C_{i k_i}$ as a product
of disjoint cycles. Observe that there is exactly one common edge
$M \cap M_i = \{e_i: (s_i\mapsto t_i)\}$, and so $s_i \in I$ is the only fixed element
in the permutation $\pi^{-1}\circ \pi_i$, {\em i.e.}, $(\pi^{-1}\circ \pi_i)(s_i) = s_i$
and $(\pi^{-1}\circ \pi_i)(s) \neq s$ for any $s\,(\neq s_i)\in I$.
Let us take $C_{i0} = (s_i)$ and let $\ell_{ij}$ be the length of the cycle $C_{ij},\,1\le j\le k_i$.
Then we must have $\ell_{ij} \ge 2$ for $0\le i\le n-1$ and $1\le j\le k_i$. Note that for a cycle
$C_{ij}$ in the disjoint cycle decomposition of these permutations, its corresponding cycle
$C_{ij}^g$ in $K_{n,n,}$ is of length $2\ell_{ij}$; half of the edges are from $M$ and the
other half are from $M_i^*$ (see Lemma~\ref{lem01}). We now try to partition $I$ into
$I^{(1)}$ and $I^{(2)}$ by analysing all the cycles $C_{ij},\,0\le i\le n-1\,\mbox{and}\,1\le j\le k_i$,
such that at least one element from both $I^{(1)}$ and $I^{(2)}$ appear in the representations of all
those cycles $C_{ij}$. If such a partition of $I$ exists, then the corresponding partition
of $M$ is given by  $M^{(1)} = \{(m\mapsto \pi(m)) : m\in I^{(1)}\}$ and
$M^{(2)} = \{(m \mapsto \pi(m)) : m\in I^{(2)}\}$, and by Lemma~\ref{lem01} we can
see that the partition of $M$ into $M^{(1)}$ and $M^{(2)}$ gives the required result.
In general, if $K_{n,n}$ possesses a perfect $1$-factorization, the difficulty is to identify
a suitable perfect matching that can help to get an acyclic edge-coloring of $K_{n,n}$ using only $n+2$
colors. Let us illustrate the technique by considering the case $K_{p,p}$, where $p$ is an odd prime.

\subsection{The case of $K_{p,p}$ for an odd prime $p$}\label{sec3p1}
This case was studied in~\cite{ba2} and we present here a slight variant of it.
Take $M_i$ to be the perfect matching corresponding to the permutation
$\pi_i: a \mapsto a+i \pmod{p}$ for $0\leq i \leq p-1$.  We can see that the decomposition
$\{M_i,\,0\le i\le p-1\}$ gives a perfect $1$-factorization of $K_{p,p}$. Now consider
$M$ to be the perfect matching given by the permutation $\pi: a \mapsto ax \pmod{p}$,
where $x$ is a generator of $\mathbb{Z}_p^{*}$. Let $y$ be the multiplicative inverse
of $x$ in $\mathbb{Z}_p^*$. Note that $order(x) = order(y) =p-1$ in $\mathbb{Z}_p^*$.
We can easily check that $M\cap M_i = \{e_i: (\frac{i}{x-1}\mapsto \frac{ix}{x-1})\}$.
We also have the following.
\begin{itemize}
\item[i)] $\pi^{-1} \circ \pi_0 = \pi^{-1} = C_{00}C_{01}$, where $C_{i0} = (0)$
represents the common edge $e_0: (0\mapsto 0)$ and $C_{01} = (1\,y\,y^2\,\cdots\,y^{p-2})$
is a cycle of length $(p-1)$ containing $1$.
\item[ii)] $\pi^{-1} \circ \pi_i = C_{i0}C_{i1}$ for $1\le i\le p-1$, where
$C_{i0} = (\frac{i}{x-1})$ represents the common edge $e_i : (\frac{i}{x-1}\mapsto \frac{ix}{x-1})$
and $C_{i1} = (0\ iy\ i(y^2+y)\,\cdots\,i(y^{p-3}+\cdots+y)\ -i)$ is a cycle of length
$(p-1)$ containing $0$.
\end{itemize}
Therefore we can get the required result with a partition of $I$ into $I^{(2)} = \{0,1\}$ and
$I^{(1)} = I\setminus I^{(2)}$. Then the corresponding partition of $M$ is given by
$M^{(1)}=M \setminus \{(0\mapsto 0), (1\mapsto x)\}$ and
$M^{(2)}= \{(0\mapsto 0), (1\mapsto x)\}$. Observe that the cycle $C_{i1}^g$ of $K_{n,n}$
corresponding to $C_{i1},\,1\le i\le (p-1),$ contains exactly one edge $(0 \mapsto 0)$
which belong to $M^{(2)}$ and  $C_{01}^g$ contains the edge $(1\mapsto x)\in M^{(2)}$.
The other $(p-2)$ edges of the cycle $C_{i1}^g,\,0\le i\le p-1,$ belong to $M^{(1)}$.
Now the final assignment of the colors is as follows:
\begin{itemize}
\item[--] the edges in $M_i^*$ are colored with $c_i$ for $i\in I$;
\item[--] the edges in $M^{(1)}$ are colored with $c_p$;
\item[--] the edges in $M^{(2)}$ are colored with $c_{p+1}$.
\end{itemize}
From the above discussion and by Lemma~\ref{lem00}, it is clear that the proposed
edge-coloring (with $p+2$ colors) of $K_{p,p}$ is proper and acyclic.

Note that the proposal in~\cite{ba2} is $M_0^* = M_0$ and $M^{(2)} = \{(1\mapsto x)\}$
and one can easily see that the result is still valid with such a choice as well.

\section{The case of $K_{p^2,\,p^2}$ for an odd prime $p \ge 5$}\label{sec4}
In this section we  provide an acyclic edge-coloring of $K_{p^2,p^2}$ with $p^2+2$ colors,
where $p$ is an odd prime $\ge 5$. We follow the general framework described in the
previous section. Accordingly we now summarize the set-up in this case.
We use elements of $\mcI = \{ (a,b) : a,b\in \mbZ \}$ for labelling the
vertices of $K_{p^2,p^2}$ on both sides. Let $x$ be a generator of $\mbZs$ and let $y$ be
its inverse. Observe that $order(x) = order(y) = p-1$ in $\mbZs$. 
We consider the following.
\begin{itemize}
\item[--] Let $M_{(a,b)}$ be the perfect matching corresponding to the permutation $\pi_{(a,b)}$
of the label set $\mcI$ defined by
\begin{equation*}
\pi_{(a,b)}\big((c,d)\big)  = \left\{
\begin{array}{ll}
(a,a+b+d) & \mbox{if}\ c=0\ \mbox{and}\ a+b+d\neq 0\\
(a+x b,0) & \mbox{if}\ c=0\ \mbox{and}\ a+b+d= 0\\
(a+c + x b,0) & \mbox{if}\ c\neq 0\ \mbox{and}\ b+d =  0\\
(a+c,b+d) & \mbox{if}\ c\neq 0\ \mbox{and}\ b+d\neq 0
\end{array}\right.
\end{equation*}
Then from~\cite{br} (with $\alpha = 1$ and $\beta = x$), we can see that the perfect
matchings $\{ M_{(a,b)},\, (a,b)\in \mcI \}$ form a perfect $1$-factorization of $K_{p^2,\,p^2}$.
\item[--] We choose the perfect matching $M$ corresponding to the permutation defined by
\begin{equation*}
\pi \big((c,d)\big) = (yc,xd);
\end{equation*}
\item[--] We choose the following partition of the label set $\mcI$:
\[\mcI^{(2)} = \{(0,1),(1,0),(z,z),(z,zx)\ \mbox{for}\ z\in \mbZs\}\ 
\mbox{and}\ \mcI^{(1)} = \mcI \setminus \mcI^{(2)}\]
Then the corresponding partition of
$M= M^{(1)} \sqcup M^{(2)}$ is given by
\begin{align*}
M^{(1)} &= \{ (c,d) \mapsto  (yc,xd) : (c,d) \in \mcI^{(1)}\},\\
M^{(2)} &= \{ (c,d) \mapsto (yc,xd) : (c,d) \in \mcI^{(2)}\}.
\end{align*}
We have $| \mcI^{(2)} | = |M^{(2)}| = 2p$ and $|\mcI^{(1)}| = |M^{(1)}| = p^2-2p$. 
\end{itemize}

Let $M_{(a,b)}^* = M_{(a,b)} \setminus M$ for $(a,b) \in \mcI$.
The edge-coloring of $K_{p^2,p^2}$ that we consider is as follows:
\begin{itemize}
\item[--] the edges in $M_{(a,b)}^*$ are colored with $c_{ap+b}$ for $(a,b) \in \mcI$;
\item[--] the edges in $M^{(1)}$ are colored with $c_{p^2}$ ;
\item[--] the edges in $M^{(2)}$ are colored with $c_{p^2+1}$.
\end{itemize}
According to the general framework discussed in the previous section, the two requirements that need to
be satisfied to establish the above edge-coloring of $K_{p^2,p^2}$ is proper and acyclic are as follows:
\begin{itemize}
\item[--] for $(a,b)\in \mcI$, there is exactly one fixed element in the permutation $\pi^{-1}\circ \pi_{(a,b)}$.
That is there is exactly one edge common to both $M$ and $M_{(a,b)}$;
\item[--] for $(a,b)\in \mcI$,
elements from both $ \mcI^{(1)}$ and $\mcI^{(2)}$ must appear in the representation of the cycles
of length $\ge 2$ in the disjoint cycle decomposition of the permutation $\pi^{-1}\circ \pi_{(a,b)}$.
\end{itemize}
Let us now prove that the above two requirements are satisfied in our set-up.
For brevity of expression, we sometimes use the following notation.
\begin{align*}
x' &= \frac{1}{x-1} = \frac{y}{1-y}\ \mbox{and}\ x_i = x+x^2 +\cdots + x^i\ \mbox{for}\ i = 1, 2,\ldots,p-2,\\
y' &= \frac{1}{y-1} = \frac{x}{1-x}\ \mbox{and}\ \,y_i = y+y^2 +\cdots + y^i\ \mbox{for}\ i = 1, 2,\ldots,p-2.
\end{align*}
\begin{proposition}\label{prop41}
For $(a,b)\in \mcI$, we have $ | M\cap M_{(a,b)} |=1$.
\end{proposition}
\begin{proof}
An edge $(c,d) \mapsto (c',d')$ is common to both $M$ and $M_{(a,b)}$ if and only if
$(c',d') = \pi_{(a,b)}\big( (c,d)\big) = \pi\big( (c,d)\big)$. Therefore by checking the four cases
\begin{align*}
(a,a+b+d) &= (0,xd)\ \quad \mbox{if}\ c=0\ \mbox{and}\ a+b+d\neq 0\\
(a+x b,0) &= (0,xd) \ \quad \mbox{if}\ c=0\ \mbox{and}\ a+b+d= 0\\
(a+c + x b,0) &= (yc,xd) \quad \mbox{if}\ c\neq 0\ \mbox{and}\ b+d =  0\\
(a+c,b+d) &= (yc,xd) \quad \mbox{if}\ c\neq 0\ \mbox{and}\ b+d\neq 0
\end{align*}
 for $(a,b)\in \mcI$, we get 
\[M\cap M_{(a,b)} = \{ (ay',bx') \mapsto (-ax',-by' )\}\]
and hence the proof.\qed\\
\end{proof}

The above proposition shows that the first requirement is satisfied in our set-up. In what follows we
prove that the other requirement is also satisfied. For this purpose, we now analyse cycle structure
of the permutations $\pi^{-1} \circ \pi_{(a,b)}$ for $(a,b)\in \mcI$.
The inverse permutation of $\pi$ is given by $\pi^{-1}\big((c,d)\big) = (xc,yd)$. So we get
\begin{equation*}
\pi^{-1} \circ \pi_{(a,b)} \big((c,d)\big)  = \left\{
\begin{array}{ll}
(xa,y(a+b+d)) & \mbox{if}\ c=0\ \mbox{and}\ a+b+d\neq 0\\
(xa+x^2 b,0) & \mbox{if}\ c=0\ \mbox{and}\ a+b+d= 0\\
(x(a+c) + x^2 b,0) & \mbox{if}\ c\neq 0\ \mbox{and}\ b+d =  0\\
(x(a+c),y(b+d)) & \mbox{if}\ c\neq 0\ \mbox{and}\ b+d\neq 0
\end{array}\right.
\end{equation*}
The above permutation can be decomposed into the following three permutations.
\begin{equation*}
\pi_{(a,b)}^{(0)} \big((c,d)\big) = (x(a+c),y(b+d))
\end{equation*}
\begin{equation*}
\pi_{(a,b)}^{(1)} \big((c,d)\big)  = \left\{
\begin{array}{ll}
(c,ya+d) & \mbox{if}\ c=xa\\
(c,d) & \mbox{otherwise}
\end{array}\right.
\end{equation*}
\begin{equation*}
\pi_{(a,b)}^{(2)} \big((c,d)\big)  = \left\{
\begin{array}{ll}
(c+x^2 b,d) & \mbox{if}\ d=0\\
(c,d) & \mbox{otherwise}
\end{array}\right.
\end{equation*}
\begin{proposition}\label{prop42}
For $(a,b)\in \mcI$, we have
\[ \pi^{-1} \circ \pi_{(a,b)} = \pi_{(a,b)}^{(2)}\circ \pi_{(a,b)}^{(1)}\circ \pi_{(a,b)}^{(0)}\]
\end{proposition}
\begin{proof}
Note that 
\[ \pi_{(a,b)}^{(1)}\big( \pi_{(a,b)}^{(0)} \big((c,d)\big) \big) = \left\{
\begin{array}{ll}
(x(a+c),ya+y(b+d)) & \mbox{if}\  x(a+c) =xa\  ( \Leftrightarrow c = 0)\\
(x(a+c),y(b+d)) & \mbox{otherwise}
\end{array}\right.
\]
and $\pi_{(a,b)}^{(2)}$ splits each of above two cases into two subcases depending on whether the
second component is zero or not. Thus we get
\[ \pi_{(a,b)}^{(2)}\big( \pi_{(a,b)}^{(1)}\circ \pi_{(a,b)}^{(0)} \big((c,d)\big)  \big)
= \left\{ \begin{array}{ll}
(xa,ya+y(b+d)) & \mbox{if}\ c=0\ \mbox{and}\ a+b+d\neq 0\\
(xa+x^2 b,0) & \mbox{if}\ c=0\ \mbox{and}\ a+b+d= 0\\
(x(a+c) + x^2 b,0) & \mbox{if}\ c\neq 0\ \mbox{and}\ b+d =  0\\
(x(a+c),y(b+d)) & \mbox{if}\ c\neq 0\ \mbox{and}\ b+d\neq 0
\end{array}\right.\]
and hence the proof.
\qed\\
\end{proof}

With the above decomposition and the following result, analysis of the cycle structures can be
simplified which we will see later.

\begin{proposition}\label{prop43}
Let $\sigma \big((c,d)\big) = (yc,yd)$, then we have
\[\pi^{-1} \circ \pi_{(xa,xb)} = \sigma^{-1} \circ ( \pi^{-1} \circ  \pi_{(a,b)}) \circ \sigma.\]
\end{proposition}
\begin{proof}
We have 
\[ \pi^{-1} \circ  \pi_{(a,b)} \big(\sigma \big((c,d)\big)\big)  = \left\{
\begin{array}{ll}
(xa,y(a+b+yd)) & \mbox{if}\ c=0\ \mbox{and}\ a+b+yd\neq 0\\
(xa+x^2 b,0) & \mbox{if}\ c=0\ \mbox{and}\ a+b+yd= 0\\
(x(a+yc) + x^2 b,0) & \mbox{if}\ c\neq 0\ \mbox{and}\ b+yd =  0\\
(x(a+yc),y(b+yd)) & \mbox{if}\ c\neq 0\ \mbox{and}\ b+yd\neq 0
\end{array}\right.\]
and so we get
\[ \sigma^{-1} \big( \pi^{-1} \circ  \pi_{(a,b)}\circ \sigma \big((c,d)\big)\big)  = \left\{
\begin{array}{ll}
(x^2 a,y(xa+xb+d) & \mbox{if}\ c=0\ \mbox{and}\ xa+xb+d\neq 0\\
(x^2 a+x^3 b,0) & \mbox{if}\ c=0\ \mbox{and}\ xa+xb+d= 0\\
(x^2a+c + x^3 b,0) & \mbox{if}\ c\neq 0\ \mbox{and}\ xb+d =  0\\
(x^2 a+c,y(xb+d)) & \mbox{if}\ c\neq 0\ \mbox{and}\ xb+d\neq 0
\end{array}\right.\]
Now one can check that this is equal to $\pi^{-1} \circ \pi_{(xa,xb)}$.
\qed\\
\end{proof}

Thus for $(a,b)\in \mcI$, the permutations $\pi^{-1} \circ \pi_{(x^i a,x^i b)},\, i\in \{0,1,\ldots,p-2\},$
are all conjugates of each other, and so they all have same cycle structure. In fact we
get the disjoint cycle decomposition of $\pi^{-1} \circ \pi_{(x^i a,x^i b)}$ by the symbol
transformation $\sigma^i$, {\em i.e.}, replacing the symbols $(c,d)$
by $(y^ic,y^id)$ in the disjoint cycle decomposition of $\pi_{(a,b)}$.
Therefore it is enough to study the cycle structure of $\pi^{-1} \circ \pi_{(a',b')}$
for $(a',b') \in \mcI' = \{(0,0),(0,1),(1,0),(1,b)\ \mbox{for}\ b\in\mbZs\}$.

We now analyse the cycle structure of $\pi^{-1} \circ \pi_{(a,b)}$ for $(a,b)\in \mcI$
by dividing into four cases: $(0,0),(*,0),(0,*),(*,*)$, where $*$ represents nonzero
elements of $\mbZ$. We discuss these cases one by one and we show that elements
from both $ \mcI^{(1)}$ and $ \mcI^{(2)}$ appear in the representation of the cycles
(of length $\ge 2$) in the disjoint cycle decomposition of $\pi^{-1}\circ \pi_{(a,b)}$
for elements $(a,b)$ in each of these four cases.

In Section~\ref{sec3} we have mentioned cycles with single element explicitly to
emphasize that there is exactly one fixed element. As we have already proved it in
Proposition~\ref{prop41}, for convenience we follow the convention and in
the discussion below we do not explicitly mention cycles with single element
in the disjoint cycle decomposition of permutations. Accordingly we count only
the cycles of length $\ge 2$ in the disjoint cycle decomposition. For simplicity,
we also use some common notation in presenting the disjoint cycle decomposition
in each of these cases.\\

\noindent{\bf Case}: $a=0\ \mbox{and}\ b=0$

The permutation $ \pi_{(0,0)}$ is the identity map, and so
$\pi^{-1} \circ \pi_{(0,0)} = \pi^{-1}$. 
We can see from the definition that the disjoint cycle decomposition of $\pi^{-1}$ can be given by
\[ \pi^{-1} = C_0 C_1 \cdots C_{p-1} C_p,\]
where the $p+1$ cycles are given by
\begin{align*}
C_j &= \big( (j,1)\, (jx,y)\, (jx^2,y^2)\cdots (jx^{p-2},y^{p-2}) \big)\quad 
\mbox{for}\ j\in \mbZ\\
C_p &=  \big( (1,0)\, (x,0)\, (x^2,0)\cdots (x^{p-2},0) \big)
\end{align*}
and it is evident that they are of length $p-1$. 
The missing element $(0,0)$ is fixed by $\pi^{-1}$.

Note that the cycle $C_0$ contains exactly one element $(0,1)\in \mcI^{(2)}$
and other elements belong to $\mcI^{(1)}$, and also the cycle $C_p$ contains exactly
one element $(1,0)\in \mcI^{(2)}$ and other elements belong to $\mcI^{(1)}$.

For $j\in \mbZs$, the elements of the cycle $C_j$ are of the form $(jx^i,y^i)$
for $i=0,1,\ldots,p-2$. Since $y$ is also a generator of $\mbZs$ we have either
$j = y^{2s}$ or $j = y^{2s+1}$ for some $s \in \{0,1,\ldots,\frac{p-1}{2}-1\}$.
Note that $x^\frac{p-1}{2} = y^\frac{p-1}{2} = -1$.
\begin{itemize}
\item[--] If $j=y^{2s}$ (is a square),
then we have $jx^s = y^s$ and $jx^{s+\frac{p-1}{2}} = y^{s+\frac{p-1}{2}}$,
and observe that $(jx^s,y^s)$ and $(jx^{s+\frac{p-1}{2}},y^{s+\frac{p-1}{2}})$
are the only elements of $C_j$ which belong to $\mcI^{(2)}$ and they are of the form $(z,z)$.
The remaining $(p-3)$ elements of $C_j$ belong to $\mcI^{(1)}$.
\item[--] If $j=y^{2s+1}$ (is a non-square), then we have
$jx^s = y\cdot y^s$ and $jx^{s+\frac{p-1}{2}} = y \cdot y^{s+\frac{p-1}{2}}$,
and observe that $(jx^s,y^s)$ and $(jx^{s+\frac{p-1}{2}},y^{s+\frac{p-1}{2}})$
are the only elements of $C_j$ which belong to $\mcI^{(2)}$ and they are of the form $(z,zx)$.
The remaining $(p-3)$ elements of $C_j$ belong to $\mcI^{(1)}$.
\end{itemize}

\noindent{\bf Case}: $a=0\ \mbox{and}\ b\in \mbZs$

By Proposition~\ref{prop43} it is enough to study the cycle structure of $\pi^{-1} \circ \pi_{(0,1)}$.
For $b\in \mbZs$, the disjoint cycle decomposition of $\pi^{-1} \circ \pi_{(0,b)}$ can be obtained
by the symbol transformation: replacing $(c,d)$ by $(bc,bd)$ in the disjoint cycle decomposition of
$\pi^{-1} \circ \pi_{(0,1)}$.

We have \[ \pi_{(0,1)}^{(0)} \big((c,d)\big) = (xc,y(1+d)).\]
Then we can see that
\[ \pi_{(0,1)}^{(0)} = C_0 C_1 \cdots C_{p-1} C_p,\]
where the $p+1$ cycles are given by
\begin{align}\label{eqnae0b}
C_j &= \big( (jx^2,0)\, (jx^3,y)\, (jx^4,y+y^2)\cdots (jx,y+\cdots+y^{p-2}) \big)\quad 
\mbox{for}\ j\in \mbZ\nonumber\\
C_p  &=  \big( (1,x')\, (x,x')\, (x^2,x')\cdots (x^{p-2},x') \big),
\end{align}
and it is evident that they are of length $p-1$. 
The missing element $(0,x')$ is fixed by $\pi_{(0,1)}^{(0)}$.

We get $\pi_{(0,1)}^{(1)} \big((c,d)\big) = (c,d)$ since $a=0$, and so $\pi_{(0,1)}^{(1)}$
is the identity map. We have
\[\pi_{(0,1)}^{(2)} \big((c,d)\big) = \left\{\begin{array}{ll}
(c+x^2,0) & \mbox{if}\ d=0\\
(c,d) & \mbox{otherwise}
\end{array}\right.\]
So it has only one cycle given by
\[ \pi_{(0,1)}^{(2)} = C^{(2)} = \big( (0,0)\,(x^2,0)\,(2x^2,0)\cdots ((p-1)x^2,0))\big),\]
it is of length $p$ and all other elements are fixed.

Therefore we have
\[  \pi^{-1} \circ \pi_{(0,1)} = \pi_{(0,1)}^{(2)}\circ \pi_{(0,1)}^{(1)}\circ \pi_{(0,1)}^{(0)}
= C^{(2)} C_0 C_1 \cdots C_{p-1} C_p.\]
Observe that an element $(kx^2,0)$ of the cycle $C^{(2)}$ appears exactly in one cycle $C_j$
(when $j = k$). So the product $C^{(2)} C_0 C_1 \ldots C_{p-1}$ will form a single cycle of length
$p(p-1)$, and Figure~\ref{fig5} shows how the cycles are joined together to form a single cycle.

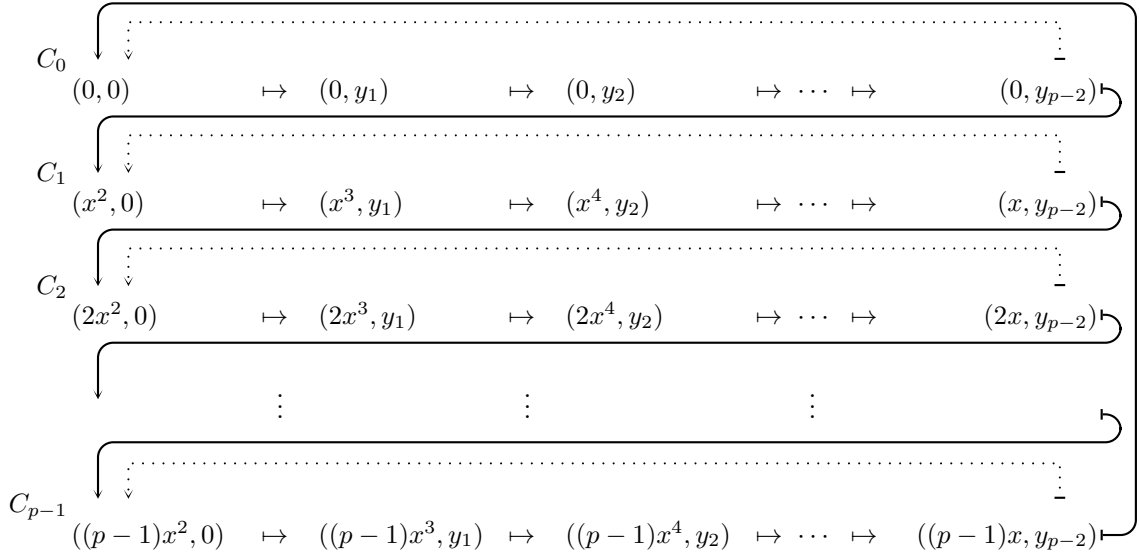
\begin{figure}[ht]
\centering{
{\small
\scalebox{1}
{
\begin{pspicture}(0,1.1)(15.6,8.1)
\rput[Br](0.95,7.4){$C_0 $}
\rput[Br](0.95,5.9){$C_1 $}
\rput[Br](0.95,4.4){$C_2 $}
\rput[Br](0.95,1.5){$C_{p-1} $}

\rput[Bl](1,7){$(0,0)$}
\rput[Bl](1,5.5){$(x^2,0)$}
\rput[Bl](1,4){$(2x^2,0)$}
\rput[Bl](1,1.1){$((p-1)x^2,0)$}

\rput[Bl](3.5,7){$\mapsto$}
\rput[Bl](3.5,5.5){$\mapsto$}
\rput[Bl](3.5,4){$\mapsto$}
\rput[Bl](3.5,1.1){$\mapsto$}
\rput[B](3.75,2.75){$\vdots$}

\rput[Bl](4.25,7){$(0,y_1)$}
\rput[Bl](4.25,5.5){$(x^3,y_1)$}
\rput[Bl](4.25,4){$(2x^3,y_1)$}
\rput[Bl](4.25,1.1){$((p-1)x^3,y_1)$}

\rput[Bl](6.75,7){$\mapsto$}
\rput[Bl](6.75,5.5){$\mapsto$}
\rput[Bl](6.75,4){$\mapsto$}
\rput[Bl](6.75,1.1){$\mapsto$}
\rput[B](7,2.75){$\vdots$}

\rput[Bl](7.5,7){$(0,y_2)$}
\rput[Bl](7.5,5.5){$(x^4,y_2)$}
\rput[Bl](7.5,4){$(2x^4,y_2)$}
\rput[Bl](7.5,1.1){$((p-1)x^4,y_2)$}

\rput[Bl](10,7){$\mapsto$}
\rput[Bl](10,5.5){$\mapsto$}
\rput[Bl](10,4){$\mapsto$}
\rput[Bl](10,1.1){$\mapsto$}
\rput[B](10.75,2.75){$\vdots$}

\rput[Bl](11.25,7){$\mapsto$}
\rput[Bl](11.25,5.5){$\mapsto$}
\rput[Bl](11.25,4){$\mapsto$}
\rput[Bl](11.25,1.1){$\mapsto$}

\rput[Bl](10.55,7){$\cdots$}
\rput[Bl](10.55,5.5){$\cdots$}
\rput[Bl](10.55,4){$\cdots$}
\rput[Bl](10.55,1.1){$\cdots$}

\rput[Br](14.5,7){$(0,y_{p-2})$}
\rput[Br](14.5,5.5){$(x,y_{p-2})$}
\rput[Br](14.5,4){$(2x,y_{p-2})$}
\rput[Br](14.5,1.1){$((p-1)x, y_{p-2})$}

\psline[linearc=0.2,tbarsize=4pt 0,arrowsize=0.4pt 4,arrowlength=1,arrowinset=0.7]{|*->}(14.55,7.12)(14.8,7.12)(14.8,6.75)(1.35,6.75)(1.35,6)
\psline[linearc=0.2,tbarsize=4pt 0,arrowsize=0.4pt 4,arrowlength=1,arrowinset=0.7]{|*->}(14.55,5.62)(14.8,5.62)(14.8,5.25)(1.35,5.25)(1.35,4.5)
\psline[linearc=0.2,tbarsize=4pt 0,arrowsize=0.4pt 4,arrowlength=1,arrowinset=0.7]{|*->}(14.55,4.12)(14.8,4.12)(14.8,3.75)(1.35,3.75)(1.35,3)
\psline[linearc=0.2,tbarsize=4pt 0,arrowsize=0.4pt 4,arrowlength=1,arrowinset=0.7]{|*->}(14.55,2.8)(14.8,2.8)(14.8,2.43)(1.35,2.43)(1.35,1.68)
\psline[linearc=0.2,tbarsize=4pt 0,arrowsize=0.4pt 4,arrowlength=1,arrowinset=0.7]{|*->}(14.55,1.2)(15,1.2)(15,1.7)(15,8.25)(1.35,8.25)(1.35,7.5)
\psline[tbarsize=4pt 0,arrowsize=0.4pt 4,arrowlength=1,arrowinset=0.7,linestyle=dotted]{|->}(14,7.5)(14,8)(13.5,8) (1.75,8)(1.75,7.5)
\psline[tbarsize=4pt 0,arrowsize=0.4pt 4,arrowlength=1,arrowinset=0.7,linestyle=dotted]{|->}(14,6)(14,6.5)(13.5,6.5)(1.75,6.5)(1.75,6)
\psline[tbarsize=4pt 0,arrowsize=0.4pt 4,arrowlength=1,arrowinset=0.7,linestyle=dotted]{|->}(14,4.5)(14,5)(13.5,5)(1.75,5)(1.75,4.5)
\psline[tbarsize=4pt 0,arrowsize=0.4pt 4,arrowlength=1,arrowinset=0.7,linestyle=dotted]{|->}(14,1.68)(14,2.15)(13.5,2.15)(1.75,2.15)(1.75,1.68)
\end{pspicture} 
}}
\caption{The cycle formed by the product $C^{(2)} C_0 C_1 \cdots C_{p-1}$}\label{fig5}
}
\end{figure}

Let $F_1$ denote the cycle formed by the product $C^{(2)} C_0 C_1 \ldots C_{p-1}$
and let $F_2 = C_p$. Thus we have \[\pi^{-1} \circ \pi_{(0,1)} = F_1 F_2\]
is a product of two disjoint cycles. Therefore for $b\in \mbZs$, the disjoint cycle
decomposition of $\pi^{-1} \circ \pi_{(0,b)}$ can be given by
\[\pi^{-1} \circ \pi_{(0,b)} = \bar{F_1} \bar{F_2},\]
where the cycles $\bar{F_1}$ and $\bar{F_2}$ are obtained by the symbol transformation:
replacing $(c,d)$ by $(bc,bd)$ in $F_1$ and $F_2$ respectively. From~(\ref{eqnae0b}) we
can see that the cycle $\bar{F_2}$ can be given by
\[\bar{F_2} =  \big( (b,bx')\, (bx,bx')\, (bx^2,bx')\cdots (bx^{p-2},bx') \big) = 
\big( (1,bx')\, (x,bx')\, (x^2,bx')\cdots (x^{p-2},bx')\big)\]
Let $bx' = x^s$ for some $s\in \{0,1,\ldots,p-2\}$ and observe that the
cycle $\bar{F_2}$ contains exactly two elements of $\mcI^{(2)}$, namely,
$(x^s,bx')$ of the form $(z,z)$ and $(x^{s-1},bx')$ of the form $(z,zx)$.
The other $(p-3)$ elements of $\bar{F_2}$ belong to $\mcI^{(1)}$. Evidently,
elements from both $\mcI^{(1)}$ and $\mcI^{(2)}$ appear in the cycle $\bar{F_1}$.\\

\noindent{\bf Case}: $a\in \mbZs\ \mbox{and}\ b=0$

By Proposition~\ref{prop43} it is enough to study the cycle structure of $\pi^{-1} \circ \pi_{(1,0)}$.
For $a\in \mbZs$, the disjoint cycle decomposition of $\pi^{-1} \circ \pi_{(a,0)}$ can be obtained
by the symbol transformation: replacing $(c,d)$ by $(ac,ad)$ in the disjoint cycle decomposition of
$\pi^{-1} \circ \pi_{(1,0)}$.

We have \[ \pi_{(1,0)}^{(0)} \big((c,d)\big) = (x(1+c),yd).\]
Then we can see that
\[\pi_{(1,0)}^{(0)} = C_0 C_1 \cdots C_{p-1} C_p.\]
where the $p+1$ cycles are given by
\begin{align}\label{eqnabe0}
C_j &= \big((x,jy)\, (x+x^2,jy^2)\cdots (x+\cdots+x^{p-2},jy^{p-2})\, (0,j)\big)
\quad \mbox{for}\ j\in \mbZ\nonumber\\
C_p &=  \big( (y',1)\, (y',y)\, (y',y^2)\cdots (y',y^{p-2}) \big),
\end{align}
and it is evident that they all are of length $p-1$.
The missing element $(y',0)$ is fixed by $\pi_{(1,0)}^{(0)}$.

We have \[\pi_{(1,0)}^{(1)} \big((c,d)\big) = \left\{\begin{array}{ll}
(c,y+d) & \mbox{if}\ c=x\\
(c,d) & \mbox{otherwise}
\end{array}\right.\]
So it has only one cycle given by
\[\pi_{(1,0)}^{(1)}= C^{(1)} = \big( (x,0)\,(x,y)\,(x,2y)\cdots (x,(p-1)y))\big),\]
it is of length $p$ and all other elements are fixed.
We get $\pi_{(1,0)}^{(2)} \big((c,d)\big) = (c,d)$ since $b=0$, and so
$\pi_{(1,0)}^{(2)}$ is the identity map.
Therefore we have
\[  \pi^{-1} \circ \pi_{(1,0)} = \pi_{(1,0)}^{(2)}\circ \pi_{(1,0)}^{(1)}\circ \pi_{(1,0)}^{(0)}
= C^{(1)}  C_0 C_1 \cdots C_{p-1} C_p.\]

Observe that an element $(x,ky)$ of the cycle $C^{(1)}$ appears exactly in one cycle $C_j$
(when $j = k$). So the product $C^{(1)} C_0 C_1 \ldots C_{p-1}$ will form a single cycle of
length $p(p-1)$, and Figure~\ref{fig6} shows how the cycles are joined together to form a single cycle.

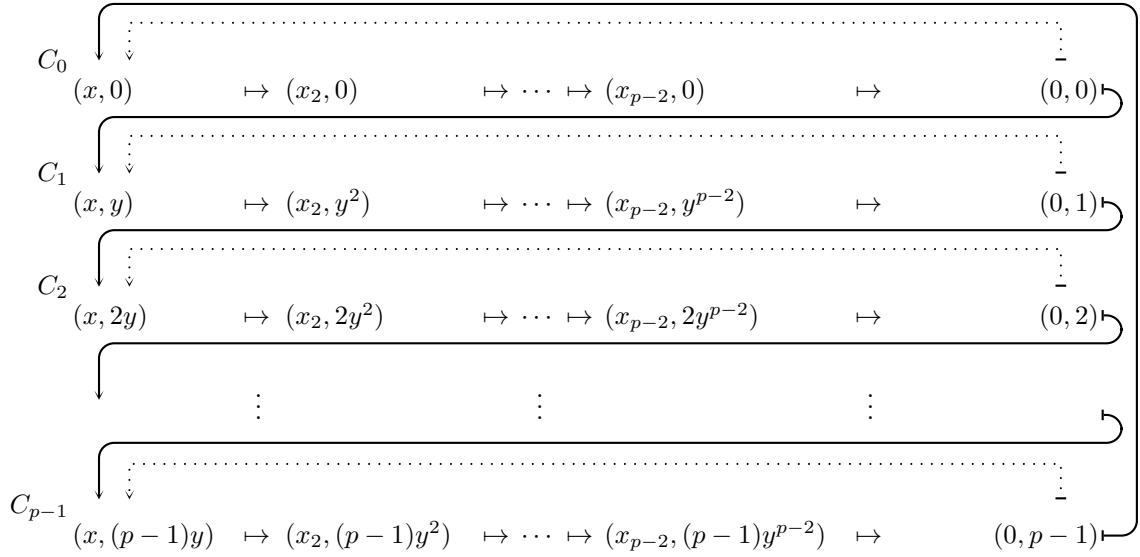
\begin{figure}[ht]
\centering{
{\small
\scalebox{1}
{
\begin{pspicture}(0,1.1)(15.6,8.1)
\rput[Br](0.95,7.4){$C_0 $}
\rput[Br](0.95,5.9){$C_1 $}
\rput[Br](0.95,4.4){$C_2 $}
\rput[Br](0.95,1.5){$C_{p-1} $}

\rput[Bl](1,7){$(x,0)$}
\rput[Bl](1,5.5){$(x,y)$}
\rput[Bl](1,4){$(x,2y)$}
\rput[Bl](1,1.1){$(x,(p-1)y)$}

\rput[Bl](3.25,7){$\mapsto$}
\rput[Bl](3.25,5.5){$\mapsto$}
\rput[Bl](3.25,4){$\mapsto$}
\rput[Bl](3.25,1.1){$\mapsto$}
\rput[B](3.45,2.75){$\vdots$}

\rput[Bl](3.8,7){$(x_2,0)$}
\rput[Bl](3.8,5.5){$(x_2,y^2)$}
\rput[Bl](3.8,4){$(x_2,2y^2)$}
\rput[Bl](3.8,1.1){$(x_2,(p-1)y^2)$}

\rput[Bl](6.4,7){$\mapsto$}
\rput[Bl](6.4,5.5){$\mapsto$}
\rput[Bl](6.4,4){$\mapsto$}
\rput[Bl](6.4,1.1){$\mapsto$}
\rput[B](7.15,2.75){$\vdots$}

\rput[Bl](6.9,7){$\cdots$}
\rput[Bl](6.9,5.5){$\cdots$}
\rput[Bl](6.9,4){$\cdots$}
\rput[Bl](6.9,1.1){$\cdots$}

\rput[Bl](7.5,7){$\mapsto$}
\rput[Bl](7.5,5.5){$\mapsto$}
\rput[Bl](7.5,4){$\mapsto$}
\rput[Bl](7.5,1.1){$\mapsto$}

\rput[Bl](8,7){$(x_{p-2},0)$}
\rput[Bl](8,5.5){$(x_{p-2},y^{p-2})$}
\rput[Bl](8,4){$(x_{p-2},2y^{p-2})$}
\rput[Bl](8,1.1){$(x_{p-2},(p-1)y^{p-2})$}

\rput[Bl](11.3,7){$\mapsto$}
\rput[Bl](11.3,5.5){$\mapsto$}
\rput[Bl](11.3,4){$\mapsto$}
\rput[Bl](11.3,1.1){$\mapsto$}
\rput[B](11.5,2.75){$\vdots$}

\rput[Br](14.5,7){$(0,0)$}
\rput[Br](14.5,5.5){$(0,1)$}
\rput[Br](14.5,4){$(0,2)$}
\rput[Br](14.5,1.1){$(0, p-1)$}

\psline[linearc=0.2,tbarsize=4pt 0,arrowsize=0.4pt 4,arrowlength=1,arrowinset=0.7]{|*->}(14.55,7.12)(14.8,7.12)(14.8,6.75)(1.35,6.75)(1.35,6)
\psline[linearc=0.2,tbarsize=4pt 0,arrowsize=0.4pt 4,arrowlength=1,arrowinset=0.7]{|*->}(14.55,5.62)(14.8,5.62)(14.8,5.25)(1.35,5.25)(1.35,4.5)
\psline[linearc=0.2,tbarsize=4pt 0,arrowsize=0.4pt 4,arrowlength=1,arrowinset=0.7]{|*->}(14.55,4.12)(14.8,4.12)(14.8,3.75)(1.35,3.75)(1.35,3)
\psline[linearc=0.2,tbarsize=4pt 0,arrowsize=0.4pt 4,arrowlength=1,arrowinset=0.7]{|*->}(14.55,2.8)(14.8,2.8)(14.8,2.43)(1.35,2.43)(1.35,1.68)
\psline[linearc=0.2,tbarsize=4pt 0,arrowsize=0.4pt 4,arrowlength=1,arrowinset=0.7]{|*->}(14.55,1.2)(15,1.2)(15,1.7)(15,8.25)(1.35,8.25)(1.35,7.5)

\psline[tbarsize=4pt 0,arrowsize=0.4pt 4,arrowlength=1,arrowinset=0.7,linestyle=dotted]{|->}(14,7.5)(14,8)(13.5,8) (1.75,8)(1.75,7.5)
\psline[tbarsize=4pt 0,arrowsize=0.4pt 4,arrowlength=1,arrowinset=0.7,linestyle=dotted]{|->}(14,6)(14,6.5)(13.5,6.5)(1.75,6.5)(1.75,6)
\psline[tbarsize=4pt 0,arrowsize=0.4pt 4,arrowlength=1,arrowinset=0.7,linestyle=dotted]{|->}(14,4.5)(14,5)(13.5,5)(1.75,5)(1.75,4.5)
\psline[tbarsize=4pt 0,arrowsize=0.4pt 4,arrowlength=1,arrowinset=0.7,linestyle=dotted]{|->}(14,1.68)(14,2.15)(13.5,2.15)(1.75,2.15)(1.75,1.68)
\end{pspicture} 
}}
\caption{The cycle formed by the product $C^{(1)} C_0 C_1 \cdots C_{p-1}$}\label{fig6}
}
\end{figure}

Let $F_1$ denote the cycle formed by the product $C^{(1)} C_0 C_1 \ldots C_{p-1}$ and let $F_2 = C_p$.
Thus we have \[ \pi^{-1} \circ \pi_{(1,0)} = F_1 F_2\]
is a product of two disjoint cycles. Therefore for $a\in \mbZs$, the disjoint cycle decomposition of
$\pi^{-1} \circ \pi_{(a,0)}$ can be given by
\[\pi^{-1} \circ \pi_{(a,0)} = \bar{F_1} \bar{F_2},\]
where the cycles $\bar{F_1}$ and $\bar{F_2}$ are obtained by the symbol transformation:
replacing $(c,d)$ by $(ac,ad)$ in $F_1$ and $F_2$ respectively. From~(\ref{eqnabe0}) we
can see that the cycle $\bar{F_2}$ can be given by
\[\bar{F_2} = \big( (ay',a)\, (ay',ay)\, (ay',ay^2)\cdots (ay',ay^{p-2})\big) =
\big( (ay',1)\, (ay',y)\, (ay',y^2)\cdots (ay',y^{p-2}) \big)\]
Let $ay' = y^s$ for some $s\in \{0,1,\ldots,p-2\}$ and observe that the
cycle $\bar{F_2}$ contains exactly two elements of $\mcI^{(2)}$, namely, $(ay', y^s)$
of the form $(z,z)$ and $(ay',y^{s-1})$ of the form $(z,zx)$. The other $(p-3)$ elements
of $\bar{F_2}$ belong to $\mcI^{(1)}$. Evidently,
elements from both $\mcI^{(1)}$ and $\mcI^{(2)}$ appear in the cycle $\bar{F_1}$.\\

\noindent{\bf Case}: $a\in \mbZs\ \mbox{and}\ b\in \mbZs$

By Proposition~\ref{prop43} it is enough to study the cycle structure of $\pi^{-1} \circ \pi_{(1,b)}$
for $ b\in \mbZs$. For $a\in \mbZs$, the disjoint cycle decomposition of 
$\pi^{-1} \circ \pi_{(a,ab)}$ (observe that all the elements of this
case are covered) can be obtained by the symbol transformation: replacing $(c,d)$ by $(ac,ad)$ in
the disjoint cycle decomposition of $\pi^{-1} \circ \pi_{(1,b)}$.

We have \[ \pi_{(1,b)}^{(0)} \big((c,d)\big) = (x(1+c),y(b+d)).\]
Then we can see that
\[ \pi_{(1,b)}^{(0)} = C_0 C_1 \cdots C_{p-1} C_p.\]
where the $p+1$ cycles are given by

\begin{align}\label{eqnab}
C_j &= \big((x,(b+j)y)\, (x+x^2,(y+y^2)b+jy^2)\cdots (x_{p-2},y_{p-2}b+jy^{p-2})\, (0,j)\big)
\ \ \mbox{for}\ j\in \mbZ\nonumber\\
C_p  &=  \big( (y',0)\, (y',yb)\, (y',(y+y^2)b)\cdots (y',y_{p-2}b \big),
\end{align}
and it is evident that they all are of length $p-1$.
The missing element $(y',bx')$ is fixed by $\pi_{(1,b)}^{(0)}$.

Note that the cycle $C_j$ for $j = bx'$ is given by
\begin{equation}\label{eqnbxp}
C_{bx'} =\big((x_1,bx')\, (x_2,bx')\cdots (x_{p-2},bx')\, (0,bx')\big)
\end{equation}
and all the elements of $C_{bx'}$ are of the form $(c,bx')$, where $c\,(\neq y')\in \mbZ$,
and also all the elements of $C_p$ are of the form $(y',d)$, where $d\,(\neq bx')\in \mbZ$.
As in the previous two cases, we will show that the permutation
\[\pi^{-1} \circ \pi_{(1,b)} = \pi_{(1,b)}^{(2)}\circ \pi_{(1,b)}^{(1)}\circ \pi_{(1,b)}^{(0)} = F_1 F_2\]
can be written as a product of two disjoint cycles. Moreover, we will show that all the
elements of the cycle $C_{bx'}$ appear in $F_1$ and all the elements of the cycle $C_p$ appear
in $F_2$, which is the crucial point in our proof.

We have \[\pi_{(1,b)}^{(1)} \big((c,d)\big) = \left\{\begin{array}{ll}
(x,y+d) & \mbox{if}\ c=x\\
(c,d) & \mbox{otherwise}
\end{array}\right.\]
So it has only one cycle given by
\[ C^{(1)} = \big( (x,0)\,(x,y)\,(x,2y)\cdots (x,(p-1)y))\big),\]
it is of length $p$ and all other elements are fixed.

Observe that an element $(x,ky)$ of the cycle $C^{(1)}$ appears exactly in one cycle $C_j$
(when $b+j = k$). So the product $C^{(1)} C_0 C_1 \cdots C_{p-1}$ will form a single cycle of
length $p(p-1)$ as shown in Figure~\ref{fig7}, and let us denote it by $\hat{C}$. Then we get
\begin{equation*}\label{cycleprod0}
\pi_{(1,b)}^{(1)}\circ \pi_{(1,b)}^{(0)}
=C^{(1)} C_0 C_1 \cdots C_{p-1} C_p = \hat{C} C_p.
\end{equation*}

\begin{figure}[ht]
\centering{
{\small
\scalebox{1}
{
\begin{pspicture}(0,0.5)(15.80,10)
\rput[Bl](1,9){$(x,b y)\ \mapsto\ (x_2,by_2)\ \mapsto\,\cdots\,\mapsto\ (x_{p-2},by_{p-2})\ \mapsto\ (0,0)$}
\rput[Bl](1,7.5){$(x,b y + y)\ \mapsto\ (x_2,by_2 + y^2)\ \mapsto\,\cdots\,\mapsto\ (x_{p-2},by_{p-2}+y^{p-2})\ \mapsto\ (0,1)$}
\rput[Bl](0.5,5.5){$(x,(b+bx'-1)y) \mapsto (x_2,by_2+(bx'-1)y^2) \mapsto\cdots\mapsto (x_{p-2},by_{p-2}+(bx'-1)y^{p-2}) \mapsto (0,bx'-1)$}
\rput[Bl](1,4){$(x,(b+bx')y)= (x,bx')\ \mapsto\ (x_2,bx') \ \mapsto\cdots\mapsto\  (x_{p-2},bx')\ \mapsto (0,bx') $}
\rput[Bl](0.5,2.5){$(x,(b+bx'+1)y) \mapsto (x_2,by_2+(bx'+1)y^2) \mapsto\cdots\mapsto (x_{p-2},by_{p-2}+(bx'+1)y^{p-2}) \mapsto (0,bx'+1)$}
\rput[Bl](1,0.5){$(x,(b+p-1)y) \mapsto (x_2,by_2+(p-1)y^2) \mapsto\cdots\mapsto (x_{p-2},by_{p-2}+(p-1)y^{p-2}) \mapsto (0,p-1)$}

\rput[Br](0.9,0.9){$C_{p-1}$}
\rput[Br](0.9,2.9){$C_{bx'+1}$}
\rput[Br](0.9,4.4){$C_{bx'}$}
\rput[Br](0.9,5.9){$C_{bx'-1}$}
\rput[Br](0.9,7.9){$C_{1}$}
\rput[Br](0.9,9.4){$C_{0}$}

\psline[linearc=0.2,linestyle=dotted,tbarsize=4pt 0,arrowsize=0.4pt 4,arrowlength=1,arrowinset=0.7]{|->}(14,0.9)(14,1.3)(13.5,1.3)(2,1.3)(2,0.9)
\psline[linearc=0.2,linestyle=dotted,tbarsize=4pt 0,arrowsize=0.4pt 4,arrowlength=1,arrowinset=0.7]{|->}(15,2.9)(15,3.3)(14.5,3.3)(2,3.3)(2,2.9)
\psline[linearc=0.2,linestyle=dotted,tbarsize=4pt 0,arrowsize=0.4pt 4,arrowlength=1,arrowinset=0.7]{|->}(11,4.4)(11,4.8)(10.75,4.8)(1.85,4.8)(1.85,4.4)
\psline[linearc=0.2,linestyle=dotted,tbarsize=4pt 0,arrowsize=0.4pt 4,arrowlength=1,arrowinset=0.7]{|->}(15,5.9)(15,6.3)(14.5,6.3)(2,6.3)(2,5.9)
\psline[linearc=0.2,linestyle=dotted,tbarsize=4pt 0,arrowsize=0.4pt 4,arrowlength=1,arrowinset=0.7]{|->}(11.5,7.9)(11.5,8.3)(11,8.3)(2,8.3)(2,7.9)
\psline[linearc=0.2,linestyle=dotted,tbarsize=4pt 0,arrowsize=0.4pt 4,arrowlength=1,arrowinset=0.7]{|->}(9.15,9.4)(9.15,9.8)(8.65,9.8)(1.8,9.8)(1.8,9.4)

\psline[linearc=0.2,tbarsize=4pt 0,arrowsize=0.4pt 4,arrowlength=1,arrowinset=0.7]{|->}(15,1.85)(15.25,1.85)(15.25,1.5)(1.5,1.5)(1.5,0.9)
\psline[linearc=0.2,tbarsize=4pt 0,arrowsize=0.4pt 4,arrowlength=1,arrowinset=0.7]{|->}(15.4,2.6)(15.65,2.6)(15.65,2.2)(1.5,2.2)(1.5,1.8)
\psline[linearc=0.2,tbarsize=4pt 0,arrowsize=0.4pt 4,arrowlength=1,arrowinset=0.7]{|->}(11.75,4.1)(12,4.1)(12,3.6)(1.5,3.6)(1.5,2.9)
\psline[linearc=0.2,tbarsize=4pt 0,arrowsize=0.4pt 4,arrowlength=1,arrowinset=0.7]{|->}(15.4,5.6)(15.65,5.6)(15.65,5.1)(1.5,5.1)(1.5,4.4)
\psline[linearc=0.2,tbarsize=4pt 0,arrowsize=0.4pt 4,arrowlength=1,arrowinset=0.7]{|->}(15,6.85)(15.25,6.85)(15.25,6.5)(1.5,6.5)(1.5,5.9)
\psline[linearc=0.2,tbarsize=4pt 0,arrowsize=0.4pt 4,arrowlength=1,arrowinset=0.7]{|->}(11.95,7.6)(12.2,7.6)(12.2,7.25)(1.5,7.25)(1.5,6.9)
\psline[linearc=0.2,tbarsize=4pt 0,arrowsize=0.4pt 4,arrowlength=1,arrowinset=0.7]{|->}(9.5,9.1)(9.75,9.1)(9.75,8.6)(1.5,8.6)(1.5,7.9)
\psline[linearc=0.2,tbarsize=4pt 0,arrowsize=0.4pt 4,arrowlength=1,arrowinset=0.7]{|->}(15,0.5)(15.75,0.5)(15.75,0.75)(15.75,9.75)(15.75,10)(1.5,10)(1.5,9.4)
\rput[Bl](8.5,1.75){$\vdots$}
\rput[Bl](8.5,6.75){$\vdots$}
\end{pspicture} 
}}
\caption{The cycle formed by the product $C^{(1)} C_0 C_1 \cdots C_{p-1}$}\label{fig7}
}
\end{figure}
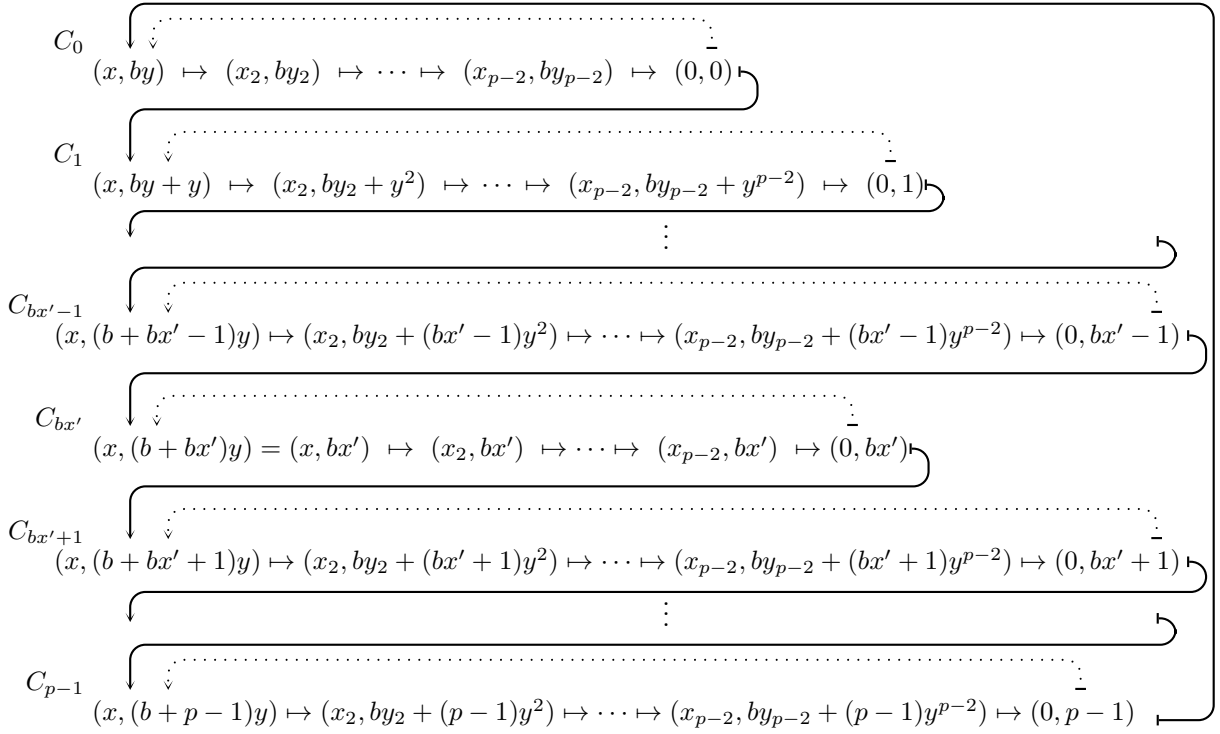

We have \[\pi_{(1,b)}^{(2)} \big((c,d)\big) = \left\{\begin{array}{ll}
(c+x^2b,0) & \mbox{if}\ d=0\\
(c,d) & \mbox{otherwise}
\end{array}\right.\]
So it has only one cycle given by
\[ C^{(2)} = \big( (0,0)\,(bx^2,0)\,(2bx^2,0)\cdots ((p-1)bx^2,0))\big),\]
it is of length $p$ and all other elements are fixed. Therefore we have
\begin{equation*}\label{cycleprod1}
\pi^{-1} \circ \pi_{(1,b)} = \pi_{(1,b)}^{(2)}\circ \pi_{(1,b)}^{(1)}\circ \pi_{(1,b)}^{(0)}
=C^{(2)} C^{(1)} C_0 C_1 \cdots C_{p-1} C_p = C^{(2)} \hat{C}C_p.
\end{equation*}

Observe that the cycle $C^{(2)}$ contains all the elements of the form $(c,0)$,
where $c\in \mbZ$. Note also that  the cycle $C_j$ for $j (\neq bx')\in \mbZ$
contains exactly one element of the form $(c,0)$ given by $(-\frac{jx}{b},0)$
(see that for some $i\in \{1,2,\ldots,p-2\}$, if the second component satisfies
$(b y_i + j y^i)=0$ then the first component $x_i = -\frac{jx}{b}$),
and the other element $(-x'x,0)=(y',0)$ appears in the cycle $C_p$. So all the elements of
$C^{(2)}$ except $(y',0)$ appear in the cycle $\hat{C} = C_0 C_1 \cdots C_{p-1}$,
and the order in which they appear can also be seen from the figure below.

\begin{figure}[H]
\centering{
{\small
\scalebox{1}
{
\begin{pspicture}(1,-1.1)(14.1,1.15)
\rput[Bl](1.2,1){$(0,0) \mapsto\,\cdots\,\mapsto (-\frac{x}{b},0) \mapsto\,\cdots\,\mapsto (-\frac{2x}{b},0)
\mapsto\,\cdots\,\mapsto (-\frac{(bx'-1)x}{b},0) = (y'+\frac{x}{b},0)$}
\rput{240}(11.75,0.2){${\mapsto\,\cdots\,\mapsto}$}
\rput[Br](14,-1){$(-\frac{(p-1)x}{b},0) \mapsfrom\,\cdots\,\mapsfrom (-\frac{(bx'+2)x}{b},0)
\mapsfrom\,\cdots\,\mapsfrom (-\frac{(bx'+1)x}{b},0) = (y' - \frac{x}{b},0)$}
\rput{120}(2.25,0.2){${\mapsto\,\cdots\,\mapsto}$}
\end{pspicture}
}}
\caption{Order of the elements of the form $(z,0)$ in $\hat{C}$}\label{fig8}
}
\end{figure}

Looking at only the elements of the form $(c,0)$ in the representation of $\hat{C}$,
observe that the element
$(-\frac{jx}{b},0)$ for $j (\neq bx')\in \mbZ$ is followed by $(-\frac{(j+1)x}{b},0)$
when $j \neq bx'-1$, and in the case when $j = bx'-1$, {\em i.e.}, $(-\frac{(bx'-1)x}{b},0)$
is followed by $(-\frac{(bx'+1)x}{b},0)$, a jump of $-\frac{2x}{b}$ in the
first component. Keeping in view how the composition is worked out,
we can see the following: suppose that an element $(c',d')$ is mapped
to $(c,0)$ in $\pi_{(1,b)}^{(1)}\circ \pi_{(1,b)}^{(0)} = \hat{C} C_p$,
then the element $(c',d')$ will be mapped to $(c+x^2 b,0)$ in 
$\pi_{(1,b)}^{(2)}\circ \pi_{(1,b)}^{(1)}\circ \pi_{(1,b)}^{(0)} = C^{(2)} \hat{C}C_p$;
the product by $C^{(2)}$ on the left with $\hat{C}C_p$ does not effect the mapping of
other elements that are not mapped to elements of the form $(c,0)$ in $\hat{C}C_p$.
Intuitively, the product by $C^{(2)}$ on the left with $\hat{C}C_p$ permutes the expressions:
from an element of the form $(c,0)$ to the next first element $(c',d')$ that is mapped to an element
of the form $(c,0)$ in the representation of $\hat{C}$ and $C_p$.
Having this in mind, let us now look at the cycles in the disjoint cycle decomposition of
$\pi_{(1,b)}^{(2)}\circ \pi_{(1,b)}^{(1)}\circ \pi_{(1,b)}^{(0)} = C^{(2)} \hat{C}C_p$.
The following summations in~(\ref{eqnt}) are useful to check the cycles
of the product $C^{(2)} \hat{C}C_p$.

Note that $x^2 b \neq \frac{x}{b}$ for any $b\in \mbZs$ since otherwise $xb^2 = 1$,
a contradiction to the fact that $x$ is a generator of $\mbZs$ and it can not be a square.
Let $t, 1\le t\le p-1,$ be the inverse of $(x b^2 -1)$ in $\mbZs$. Then we can see that $t$
is the smallest positive integer satisfying
\begin{equation}\label{eqnt}
 t x^2 b - (t+1) \frac{x}{b} = 0\quad \mbox{and}\quad (p-t) x^2 b - (p-t-1) \frac{x}{b} = 0.
\end{equation}

Let $F_1$ be the cycle containing the element $(-\frac{(bx'-1)x}{b},0) = (y'+\frac{x}{b},0)$ in the
product $C^{(2)} \hat{C}C_p$ and it can be seen from the figure below. Clearly, the cycle $F_1$
contains $t$ elements of the form $(c,0)$.

\begin{figure}[ht]
\centering{
{\small
\scalebox{1}
{
\begin{pspicture}(0,-1)(16,1.5)
\rput[Bl](0.6,1){$(y'+\frac{x}{b},0)\,{\mapsto\,\cdots\,\mapsto}
\overbrace{(x,bx'){\mapsto\,\cdots\,\mapsto} (x_{p-2},bx')\mapsto (0,bx')}^{\mbox{elements\ of}\ C_{bx'} }
{\mapsto\,\cdots\,\mapsto}  (y'+\frac{x}{b}-\frac{2x}{b} + x^2 b,0)$}
\rput{300}(13.4,0.1){${\mapsto\,\cdots\,\mapsto}$}
\rput[Br](15,-1){$(y'+\frac{x}{b} - \frac{tx}{b}+(t-1)x^2 b,0) \mapsfrom\,\cdots\,\mapsfrom
(y'+\frac{x}{b}-\frac{4x}{b} + 3x^2 b,0) \mapsfrom\,\cdots\,\mapsfrom
(y'+\frac{x}{b}-\frac{3x}{b} + 2x^2 b,0)$}
\rput{120}(1.75,0.1){${\mapsto\,\cdots\,\mapsto}$}
\end{pspicture}
}}
\caption{The cycle $F_1$ containing the element $(y'+\frac{x}{b},0)$ in the product
$C^{(2)} \hat{C}C_p$}\label{fig9}
}
\end{figure}
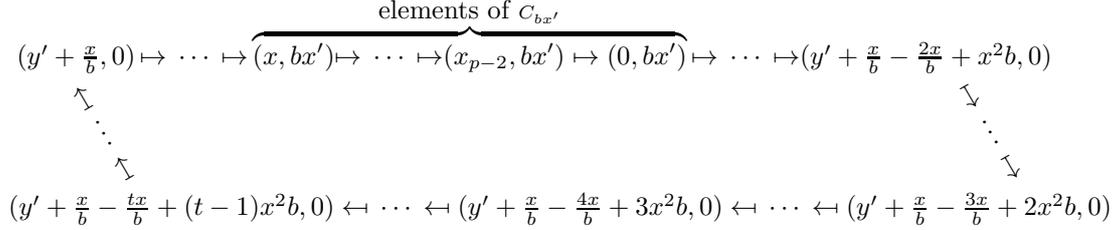

Let $F_2$ be the cycle containing the element $(y',0)$ in the product $C^{(2)} \hat{C}C_p$
and it can be seen from the figure below. Clearly, the cycle $F_2$ contains the other $(p-t)$
elements of the form $(c,0)$.

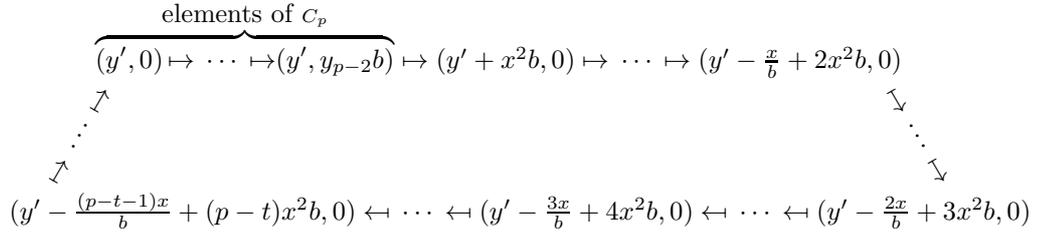
\begin{figure}[ht]
\centering{
{\small
\scalebox{1}
{
\begin{pspicture}(0,-1)(15.5,1.75)
\rput[Bl](2.7,1){$\overbrace{(y',0)\,{\mapsto\,\cdots\,\mapsto} (y',y_{p-2}b)}^{\mbox{elements\ of}\ C_{p}}
\mapsto (y'+x^2 b,0) \mapsto\,\cdots\,\mapsto (y'-\frac{x}{b}+2 x^2 b,0)$}
\rput{300}(13.5,0.1){${\mapsto\,\cdots\,\mapsto}$}
\rput[Br](15,-1){$(y' - \frac{(p-t-1)x}{b}+(p-t)x^2 b,0) \mapsfrom\,\cdots\,\mapsfrom
(y'-\frac{3x}{b} + 4x^2 b,0) \mapsfrom\,\cdots\,\mapsfrom
(y'-\frac{2x}{b} + 3x^2 b,0)$}
\rput{60}(2.5,0.1){${\mapsto\,\cdots\,\mapsto}$}
\end{pspicture}
}}
\caption{The cycle $F_2$ containing the element $(y',0)$ in the product
$C^{(2)} \hat{C}C_p$}\label{fig10}
}
\end{figure}

Note that all the elements of the form $(c,0)$ appear in the two cycles $F_1$ and $F_2$,
and so these cycles must contain all the elements of $\mcI$ except $(y',bx')$ which is the fixed
element in $\pi_{(1,b)}^{(1)}\circ \pi_{(1,b)}^{(0)} = \hat{C} C_p$. The element $(y',bx')$
is also fixed by $\pi_{(1,b)}^{(2)}$. Therefore we get
\[\pi_{(1,b)}^{(2)}\circ \pi_{(1,b)}^{(1)}\circ \pi_{(1,b)}^{(0)}
= C^{(2)} \hat{C}C_p = F_1 F_2.\]
Note also that $F_1$ contains all the elements of $C_{bx'}$ and $F_2$ contains all
the elements of $C_p$. Finally, for $a\in \mbZs$, we get that \[\pi^{-1}\circ \pi_{(a,ab)} =
\pi_{(a,ab)}^{(2)}\circ \pi_{(a,ab)}^{(1)}\circ \pi_{(a,ab)}^{(0)} = \bar{F}_1 \bar{F}_2,\]
where the cycles $\bar{F}_1$ and $\bar{F}_2$ are obtained by the symbol transformation:
replacing $(c,d)$ by $(ac,ad)$ in the cycles $F_1$ and $F_2$ respectively.

From (\ref{eqnab}) we can see that the cycle $\bar{F}_2$ contains elements of the
form $(ay',d)$ for all $d\in \mbZ$ except may be $d = abx'$. So it must contain at least
one element of the form $(z,z)$ or $(z,zx)$ which belong to $\mcI^{(2)}$.
Since $F_2$ contains all the elements of $C_p$ in the order as mentioned in~(\ref{eqnab}),
in the corresponding part of the cycle $\bar{F}_2$, at most three elements of $\mcI^{(2)}$
can appear; the other elements (at least $p-3$) belong to $\mcI^{(1)}$.
Evidently, the cycle $\bar{F}_1$ contains elements from both $\mcI^{(1)}$
and $\mcI^{(2)}$. Hence the proof is completed.\\

\noindent{\bf Acknowledgement}: The first author would like to thank Dr. T. Karthick, ISI-Chennai Centre
for suggesting the topic and helpful discussions.

\end{document}